\documentclass[11pt,a4paper,final]{article}
\usepackage[utf8]{inputenc}
\usepackage{amsmath}
\usepackage{amsfonts}
\usepackage{graphicx}
\usepackage{xcolor}
\usepackage{natbib}
\usepackage{amssymb}
\usepackage[left=1in,right=1in,top=1in,bottom=1in]{geometry}

\newcommand{\qed}{\hfill$\rule{2mm}{3mm}$}
\newenvironment{proof}{\par{\noindent \bf Proof }}{\qed \par}

\DeclareMathOperator{\opt}{optimize}

\newcommand{\xhdr}[1]{\paragraph{\bf {#1}}}

\usepackage[suppress]{color-edits}
\addauthor[Karen]{kl}{purple}
\addauthor[Hoda]{hh}{red}
\addauthor[Solon]{sb}{blue}
\addauthor[Jon]{jk}{magenta}

\newcommand{\hh}[1]{\textcolor{cyan}{[HH: #1]}}

\usepackage{amsmath}
\usepackage{amsfonts}
\usepackage{graphicx}
\usepackage{subcaption}
\usepackage{multirow}
\usepackage{wrapfig}
\usepackage{url}
\usepackage{color}

\newtheorem{theorem}{Theorem}
\newtheorem{proposition}{Proposition}

\newtheorem{lemma}{Lemma}
\newtheorem{corollary}{Corollary}

\providecommand{\floor}[1]{ \lfloor #1 \rfloor }

\newcommand{\reals}{\mathbb{R}}

\newcommand{\vp}{\mathbf{p}}

\newcommand{\vt}{\mathbf{t}}

\newcommand{\vx}{\mathbf{x}}

\newcommand{\ones}{\textbf{1}}
\newcommand{\zeros}{\textbf{0}}

\newcommand{\hl}[1]{\textcolor{black}{#1}}

\usepackage[utf8]{inputenc}

\DeclareFontFamily{U}{mathb}{\hyphenchar\font45}
\DeclareFontShape{U}{mathb}{m}{n}{
<-6> mathb5 <6-7> mathb6 <7-8> mathb7
<8-9> mathb8 <9-10> mathb9
<10-12> mathb10 <12-> mathb12
}{}
\DeclareSymbolFont{mathb}{U}{mathb}{m}{n}
\DeclareMathSymbol{\llcurly}{\mathrel}{mathb}{"CE}
\DeclareMathSymbol{\ggcurly}{\mathrel}{mathb}{"CF}

\title{On Modeling Human Perceptions of \\Allocation Policies with Uncertain Outcomes}

\date{}

\author{
   \makebox[.45\textwidth]{Hoda Heidari}\\
   Carnegie Mellon University\\
   \url{hheidari@cmu.edu} \\
   \and
   \makebox[.45\textwidth]{Solon Barocas}\\
   Microsoft Research,  Cornell University\\
   \url{solon@microsoft.com} \\
   \and
    \makebox[.45\textwidth]{Jon Kleinberg}\\
   Cornell University\\
   \url{kleinberg@cornell.edu} \\
   \and
   \makebox[.45\textwidth]{Karen Levy}\\
   Cornell University\\
   \url{karen.levy@cornell.edu} \\
}

\begin{document}
\maketitle

\begin{abstract}
\noindent Many policies allocate harms or benefits that are {\em uncertain} in nature: they produce distributions over the population in which individuals have different probabilities of incurring harm or benefit. Comparing different policies thus involves a comparison of their corresponding probability distributions, and we observe that in many instances the policies selected in practice are hard to explain by preferences based only on the expected value of the total harm or benefit they produce. In cases where the expected value analysis is not a sufficient explanatory framework, what would be a reasonable model for societal preferences over these distributions? Here we investigate explanations based on the framework of {\em probability weighting} from the behavioral sciences, which over several decades has identified systematic biases in how people perceive probabilities. We show that probability weighting can be used to make predictions about preferences over probabilistic distributions of harm and benefit that function quite differently from expected-value analysis, and in a number of cases provide potential explanations for policy preferences that appear hard to motivate by other means. In particular, we identify optimal policies for minimizing perceived total harm and maximizing perceived total benefit that take the distorting effects of probability weighting into account, and we discuss a number of real-world policies that resemble such allocational strategies. Our analysis does not provide specific recommendations for policy choices, but is instead fundamentally interpretive in nature, seeking to describe observed phenomena in policy choices. 
\end{abstract}

\section{Introduction}


Societies frequently wrestle with tough decisions regarding the allocation of benefits or burdens among their populations (see, e.g.,~\citep{calabresi1978tragic,viscusi2018pricing}). These decisions---particularly those that involve harm---are immensely difficult yet often unavoidable. As \citet{sunstein2003hazardous} points out, governments regularly pursue policies that lead to harms, including death, among the public: 
\begin{quote}
    {\em If government allows new highways to be built, it will know that people will die on those highways; if government allows new power plants to be built, it will know that some people will die from the resulting pollution. [...] Of course it would make sense, in most or all of these domains, to take extra steps to reduce risks. But that proposition does not support the implausible claim that we should disapprove, from the moral point of view, of any action taken when deaths are foreseeable.}
\end{quote}
These considerations remain true even when the prospective harms are reduced as much as possible; to the extent that they haven't been eliminated altogether, we must reason about the impact of policies that produce foreseeable harms.

To make matters more complicated, many of these allocations deal in \emph{probabilities} of some outcome occurring: when we raise the speed limit by a certain amount, for example, we can estimate to some approximate level the number of additional traffic fatalities that will result \citep{farmer2019effects}, but we can say much less about who in particular will die.
Thus, for matters involving harm, the policy process necessarily involves a set of choices (even if these choices arise only implicitly) between different {\em distributions} of harm over the population.\hhcomment{Above, we are effectively equating statistical predictions about a policy with probabilities of harm experienced by individuals. Should we clarify that in the absence of further information pinpointing who exactly is at high risk, we assume people perceive statistical facts as probabilities?}\klcomment{I don't think it's essential to clarify, but feel free if you would prefer!}\sbcomment{I also don't think we need to clarify, as doing so would interrupt the flow of the text. I think the plan was to address this point later in the text anyway.} For example, policy $P$ might produce a probability $p_i$ that individual $i$ is harmed,\hhcomment{Later we are using $\pi$ to refer to allocation policies, but it's not a big deal I think.}\klcomment{I agree but if you want to make consistent, feel free} while policy $Q$
might produce a probability $q_i$ that individual $i$ is harmed, for each individual in the population.  (To keep the discussion simple, we will think about a single kind of ``harm'' that can befall people as a result of the policy, rather than adding the complexity of different degrees of harm.)

\xhdr{Comparing probability distributions resulting from different policies} 
How should we compare the two distributions of harm that 
arise from policies $P$ and $Q$, respectively?
Much of the work that underpins mathematical models in these domains, including
many of 
the loss functions that go into algorithmic decisions, tend to be based on
expected cost---the idea that we should favor the policy that
produces the lower expected harm.
In our case, policy $P$ produces a sequence of probabilities
$(p_1, p_2, ..., p_n)$ over the $n$ members of the population, 
and its expected harm is the sum $p_1 + p_2 + \cdots + p_n$;
we can write a similar expression for the probabilities of harm
$(q_1, q_2, ..., q_n)$ produced by policy $Q$.

Of course, real-life policymaking is complex, and it is not clear that minimization of expected harm is typically the chief criterion in selecting among policy options. But there is also a more basic problem with 
using expected harm as the criterion: 
many policy questions about competing distributions of harm
begin after we've already reduced the total amount of harm to a roughly fixed, low target level,\hhcomment{should we mention here that the target level is often close to the minimum expected harm achievable given the constraints on hand?}\sbcomment{I think this point might be worth making because it helps to emphasize that we're not interested in cases where policy-makers have set the total harm at some arbitrarily level; we're interested in hard cases where the harm has---presumably---been reduced as much as we believe it can be.}
and so the debate is between distributions that all have the same
expected level of harm. How, then, should we think about preferences between 
these competing policy proposals?



\xhdr{An example: The Selective Service System}
We can see the outlines of such debates in a number of settings
where a risk of harm is being allocated across a population.
In the policies for drafting people into the military in the United States, for example, the government has considered a number of different implementations for randomizing the selection of inductees.  (Here, required service in the military is the cost, or harm, that is being allocated according to a probability distribution.)
Under a given policy $P$, individual $i$ would learn that they had a
probability $p_i$ of being drafted. 
Crucially, difficult questions about the implementations of draft systems persist regardless of the desired {\em size} of the military; that is, for a given size of the military,
the sum of the draft probabilities $p_i$ over the population is
pinned to this number, but some distributions
of these probabilities have been nonetheless viewed as preferable to others.

What accounts for these preferences? We note that discussions of revisions to the draft framed uncertainty itself as a cost being borne
by members of the population.
As the U.S. Selective Service System notes, prior to the introduction
of a structured process for randomization, men\footnote{Under U.S. law, only men are, or have ever been, required to register for the draft.} knew only that they
were eligible to be drafted from the time they turned 18 until they reached age 26; 
{\em ``[this] lack of a system resulted in uncertainty for the potential
draftees during the entire time they were within the draft-eligible
age group. All throughout a young man’s early 20’s he did not know if
he would be drafted.''} \citep{sss} The systems that were subsequently introduced specified priority groups according to age, which had the effect of deliberately producing non-uniform probabilities of being drafted; under these systems, some people learned that their probability of being drafted
was higher than average, and others learned that their probability was lower than average.\footnote{Specifically, men were drafted according to ``priority year,'' with the youngest men being drafted first. During the year a man was 20 years old, he was in the top priority group, with reduced likelihood of being called up each subsequent year. Within each group, call-up order was randomized by lottery according to birthday \citep{sss}.} Viewed in terms of distributions, these policy changes had the effect of
{\em concentrating} the 
probabilities more heavily on a subset of the eligible population, rather than
{\em diffusing} the probabilities more evenly across everyone.
The quote from the Selective Service System points out that a
process that diffuses probabilities too widely seems to create unnecessary (and harmful) levels of uncertainty; 
but there are, of course, corresponding objections that could be raised
to processes that concentrate probabilities too heavily on too small a group.

An abstraction of these questions would therefore consider multiple 
probability distributions of harm---for example, policy $P$ producing
$(p_1, p_2, ..., p_n)$, policy $Q$ producing $(q_1, q_2, ..., q_n)$, 
and perhaps others---and ask which of these should be preferred as a choice for society.
In posing such questions, we are guided by the belief that studying reactions to distributions of harm should draw closely on those parts of the behavioral sciences that have considered how people subjectively evaluate probabilities. We therefore develop a framework based on 
the concept of {\em probability weighting} from behavioral economics.

Our model will allow us to evaluate the Selective Service System's argument, and similar arguments in other domains, at a broad level---the contention that completely uniform randomization over the draft-eligible population is a sub-optimal policy because the cumulative level of uncertainty felt by the population is unnecessarily high. 
At first glance, this argument is counter-intuitive: since the size of the military is the same under all the draft policies being considered, isn't the cumulative level of uncertainty felt by the population also the same under all policies?
On closer inspection, though, we find that this decision---to shift the probabilities in a non-uniform direction, and to interpret this as reducing cumulative uncertainty---is very much consistent with the predictions of probability weighting.

\xhdr{A stylized example}
To further motivate the models that follow, we can adapt our discussions
about harm allocations---and complex scenarios such as the military draft---into a stylized example in which a fixed amount of harm must be allocated
across a given population.
We will argue that different allocations of harm have very different
subjective resonances, and it is these differences that behavioral theories
of probability weighting aim to illuminate.


Thus, as a thought experiment, consider the following hypothetical example. Suppose we need to allocate 1 unit of harm among 100 individuals. For simplicity, let's assume all 100 individuals are equally deserving and willing to bear the harm. We might allocate the harm to one specific person (say, Bob), while giving the other 99 people certainty that they are not at risk---hence the probability distribution $(1,0,\cdots, 0)$. Feeling sorry for Bob, we might instead divide the risk between him and another member of the population, Chloe---and ultimately flip a coin to decide which of them is to bear the harm\hhcomment{This might be a good place to emphasize that we are evaluating these policies ex-ante--before the risks turn into realized/deterministic outcomes.}\klcomment{I think it may be clear enough, but feel free to add if you like! If so might want to add at beginning or end of the para to avoid disrupting the flow.}\sbcomment{I agree that it is clear as is, so don't think there's any need for additions}, while the other 98 people are free and clear; i.e. the distribution $(1/2, 1/2, 0, \cdots, 0)$. Or we could have a third person, David, join Bob and Chloe in the risk pool, lowering the risk for each of them to one-third $(1/3, 1/3, 1/3, 0, \cdots, 0)$. Finally, we might allocate the risk evenly among all 100 individuals, and select the recipient of the harm by random lottery: $(0.01, \cdots, 0.01)$.

How might a policymaker select among these policies? Each of them, ultimately, results in the same amount of harm (1 unit) befalling the population, yet they strike us as intuitively quite different. We may consider it blatantly unfair to single Bob out as a certain victim by concentrating the risk completely on him; and indeed, a long line of work in psychology on the so-called {\em identifiable victim effect} suggests that we tend to find such outcomes particularly troubling \citep{jenni1997explaining}.\footnote{Philosophy has also grappled with the observation that we tend to recoil at the idea of, for example, harvesting one person's organs to save the lives of five other people. Such cases reveal an intuitive distaste for distributions that aim to reduce the overall amount of harm experienced by a population by focusing those harms on a small subset of people \citep{thomson1976killing}.
Note that our framework does not apply to these cases because concentrating costs in these instances actually reduces the total cost (e.g., reducing the total number of deaths from five to one); in our settings, the way a policy allocates harms does not affect the amount of harm imposed on the overall population.} 

On the other hand, a random lottery distributes the risk equally among all 100 individuals---but in the interim, it forces \textit{everybody} to worry about their chances of being harmed.  (This is the form of uncertainty, and corresponding psychological cost, that the Selective Service System was concerned with in our example of the draft lottery.) The second and third options provide intermediate alternatives. In the second alternative, no one person is harmed with \emph{certainty}, while at the same time, the smallest possible number of individuals need bear the risk.

The fact that we may prefer some of the above alternatives to others immediately suggests that a cost-benefit analysis based on expected harm is not sufficient to capture our intuitions---since all the options involve the same expected amount of harm.  Likewise, our intuitive reactions to these different proposals do not neatly map onto common concerns with distributive justice, where we tend to worry about the relative impact of allocations on different \hl{social groups or subgroups within} the population, given existing social inequalities. In this case, our reactions have nothing to do with any details about who Bob, Chloe, and Dave happen to be or the social groups to which they belong. What we perceive to be the more desirable allocation instead seems to rest on how we perceive the benefits or harms of being subject to uncertain outcomes.\footnote{\hl{To put it differently, the purpose of our work is \emph{not} to argue that probability-weighting tends to result in distributions that disproportionately harm members of specific social groups. Rather, we study human perceptions toward distributions that allocate the same type of harm unevenly across \emph{otherwise-equal} individuals (without specifying their group memberships). As we show later in the paper, behavioral principles suggest why people have strong reactions to uneven distributions even when the specific people who are harmed by or benefit from the allocation do not belong to distinct social groups.
}}

\xhdr{An interpretive analysis}

Our intention in exploring people's subjective perceptions of risk probabilities is, emphatically, \textit{not} to prescribe a ``best'' mode of allocating probabilities of risk, nor to endorse the underlying policy decisions that give rise to a need to allocate such risk in the first place, nor to treat superficially the variety of other procedural and moral concerns that attend the allocation of harms and benefits to people. Ours is a purely {\em interpretive} undertaking; we find that preferences for certain allocation policies involving probabilities are difficult to explain unless we take probability weighting into account. 

Policy experts disagree about the extent to which cognitive errors ought to be explicitly incorporated into account in public decision-making. While some consider it foolish to base policies on what are essentially misunderstandings, others suggest that we might reasonably consider the ``psychic benefits'' to the public of protecting against ``imaginary'' risks \citep{schneier2008psychology,viscusi2018pricing,portney1992trouble,pollak1998imagined}. We stake no claim in this debate; our goal is to explore descriptively how people's subjective perceptions of probabilities \textit{might} impact preferences regarding such allocations---and how these impacts potentially explain peculiar real-life allocation policies. In this way, our work follows a style of research that seeks to shed light on observed policy outcomes by linking them to our behavioral understanding of latent human preferences for certain types of outcomes over others (see, e.g., \citep{srivastava2019mathematical,zhu2018value,Lee2018WeBuildAI} for earlier work in this genre).

All of this still leaves us with a basic question. We have seen examples so far (with others to come) of policy-making favoring some level of randomization, while also steering away from completely uniform randomization that would spread risk of harm diffusely across a population. Is there a model that predicts this type of ``intermediate'' position that avoids both a concentration of risk on identifiable victims as well as too diffuse a distribution over the whole population? And can such a model be derived from known psychological models of human behavior? In this work we will argue that a preference for these types of intermediate distributions of risk can be derived naturally from the concept of {\em probability weighting}, one of the most empirically well-grounded human biases studied in behavioral economics~\citep{kahneman2013prospect}, to which we now turn.

\hl{\xhdr{Plan for the remainder of the paper}
Motivated by the premise that understanding people’s perceptions of harm/benefit allocations are crucial in designing acceptable policies, we posit that models that solely rely on expected-value comparisons may miss crucial aspects of human perceptions toward uncertain allocations---which are in part shaped by probability weighting. As a result, expected-value optimizing algorithms may produce allocations that are behaviorally repugnant to people. Our model can partially explain these reactions using one of the fundamental principles in behavioral sciences. To our knowledge, our work is the first to explore the attractiveness of different uncertain allocation policies by exploring \emph{optimal allocations} under probability weighing. We make several connections between the optimal allocation patterns suggested by our theory and real-world policy choices that would be otherwise difficult to explain. 
}
%
\hl{The rest of the paper is organized as follows: in Section~\ref{sec:pw}, we provide an overview of theoretical models of probability weighting as they relate to our work. Section~\ref{sec:model} introduces our new optimization-based framework, which utilizes probability weighting to understand preferences toward policies that lead to uncertain allocations. In In Sections~\ref{sec:harm_analysis} and \ref{sec:benefit_analysis}, respectively, we explore the implications of probability weighting in preferences toward distributions of harm and benefit, and discuss several real-world policy choices that align with the patterns suggested by our theory. We conclude in Section~\ref{sec:conclusion} with a brief summary and several directions for future work.
}

\section{Probability Weighting}\label{sec:pw}

\xhdr{A model based on probability weighting}
Probability weighting begins from the qualitative observation that people tend to overweight small probabilities---behaving as though they are larger than they actually are---and tend to underweight large probabilities---behaving as though they are smaller than they actually are.
More generally, probability weighting is the premise that when faced with an uncertain event of probability $p$, people will tend to behave with respect to this event---for example, when determining risks or evaluating gambles involving the event---as though its probability were not $p$ but a value $w(p)$, the {\em weighted version} of the probability.
This weighting function $w(p)$ has the two properties noted above: that $w(p)$ is larger than $p$ when $p$ is small, and $w(p)$ is smaller than $p$ when $p$ is large.
If we think in terms of the graph of $w(p)$ as a function of $p$, people refer to these properties as the ``\hl{inverse} S-shaped'' nature of the probability weighting curve. 
There are a number of different models that derive \hl{inverse} S-shaped probability weighting curves from simple observations; one influential functional form was provided by \citet{prelec1998probability}, who derived it from a set of underlying axioms about preferences for different types of gambles. (See Figure~\ref{fig:wp} for several examples of Prelec's S-shaped probability weighting functions.)
\hl{The concept of probability weighting has been invoked to explain a number of peculiar behavioral patterns; one of the canonical examples is people's participation in gambling and lotto games \citep{quiggin1991optimal,kahneman2011thinking}.\footnote{To elaborate further on this connection, note that the cost of buying a lotto ticket is always set to be higher than the expected benefit (i.e., the likelihood of winning times the prize), otherwise lottery owners would lose money. Nonetheless, people participate in these games in large numbers. Work in behavioral economics has advanced probability weighting as one explanation for this irrational behavior, via the tendency to over-weigh small probabilities---here the chance of winning the lottery \citep{quiggin1991optimal,kahneman2011thinking}.}}

We use probability weighting here to ask the following basic question.
Suppose there are $r$ units of harm to be allocated across a population of $n$ people, and we are evaluating policies that assign individual $i$ a probability $p_i$ of receiving harm, subject to the constraint that the sum of $p_i$ over all individuals $i$ is $r$.
In the motivating settings discussed so far, it is natural to think of the cost borne by individual $i$ as the perceived probability $w(p_i)$---either because individual $i$ perceives it this way (via the psychological cost of their own uncertainty) or because the rest of society views it this way (via our discomfort at the idea that $i$ is an identifiable victim with a perceived probability $w(p_i)$ of being harmed).
We can therefore ask: which probability distribution minimizes this total cost, the sum of $w(p_i)$ over all individuals $i$?
Notice that this question allows for distinctions among probability distributions that all produce the same total expected harm for the population: in particular, all the distributions under consideration have a total expected harm of $r$, but they can nevertheless differ substantially in the sum of $w(p_i)$ over all individuals $i$.

We find that the distributions minimizing the weighted sum of harm probabilities $w(p_i)$ in fact correspond to intermediate distributions of the type we have been discussing qualitatively: distributions that concentrate the risk on a subset of the population, such that each member of the at-risk subset has a probability of harm that is strictly less than $1$, 
while most of the population has a probability of harm equal to $0$.
The analysis leading to this conclusion involves some subtlety: sums of $S$-shaped functions do not exhibit the nice properties that simpler function classes do, and so minimizing them requires additional complexity in the analysis.

With this model in place, we can also explore the natural complement to this dynamic.
Our discussion thus far has focused on probabilities of {\em harm}, but there is an analogous class of questions about distributing probabilities of {\em benefit} across a population---for example, in the availability of opportunities like higher education or financial assistance programs.
Suppose there are $r$ units of benefit available to the population as a whole, and we are considering policies that assign a probability $p_i$ that individual $i$ receives the benefit.
Which distributions maximize the sum of $w(p_i)$ over all individuals $i$---that is, maximizing the total perceived benefit?  As with risks of harm, we do not argue that such a policy is necessarily desirable, only that it may have added or diminished attractiveness in its perceived impact; to the extent that such policies are favored in practice, the theory of probability weighting might therefore offer a suggestive description.

We find that the distributions maximizing this sum of perceived probabilities of benefit are quite different from the distributions minimizing the sum of perceived probabilities of harm.
In particular, when the total available benefit $r$ is small relative to the size of the population under consideration, the maximizing distribution is a uniform lottery which assigns all $n$ people a probability of $r/n$; but as $r$ increases, the maximizing distribution changes abruptly to one in which a subset of the population receives a portion of the benefit with certainty, and the rest of the population is given a uniform lottery for the remainder.


\begin{table*}
\begin{center}
    
    \begin{tabular}{|  l | p{0.35\textwidth}  | p{0.35\textwidth}  |}
    \hline
     &  Min. perceived harm &  Max. perceived benefit  \\ \hline
     $r = 1$  & $(\delta, 0.5(1-\delta), 0.5(1-\delta), 0, \cdots, 0)$ for some $\delta < \min\{\ell, 1-2\ell\}$ & $\left(\frac{1}{n}, \frac{1}{n}, \cdots, \frac{1}{n}\right)$ for all $n$ larger than a constant $q$ \\ \hline 
$r>1$ & $\left(\delta, \frac{r-\delta}{k}, \frac{r-\delta}{k}, \cdots, \frac{r-\delta}{k}, 0, \cdots, 0 \right)$ for $k<n$, $\delta < \min\{\ell, r-k\ell\}$ & $(\frac{r-j-\gamma}{n-j-1}, \cdots, \frac{r-j-\gamma}{n-j-1}, \gamma, \underbrace{1, \cdots, 1}_{j \text{ persons}})$ for $\gamma > \max\{\ell, r-j-\ell(n-j-1)\}$ \\ 
\hline
    \end{tabular}
    \caption{The summary of our findings regarding the allocation of probabilities across $n$ individuals for a probability weighting $w(.)$ that is monotone, concave up to an inflection point $\ell$, and convex beyond $\ell$.}
\label{table:summary}
\end{center}
\end{table*}

\xhdr{Implications of probability weighting}
Given that a society developing policy seems to favor some probability distributions of harm or benefit over others, even when they have the same expected value, it is natural to ask whether a model based on probability weighting can shed light on the nature of these preferences.
Our modeling activity thus works out what the favored policies would look like if society were seeking to maximize or minimize the total weighted probability.
As we discuss in Sections~\ref{sec:harm_examples} and \ref{sec:benefit_examples}, properties of these minimizing and maximizing distributions \emph{can} be observed in a variety of real-world settings. We consider a number of allocation policies that have been adopted in practice that involve distributions of uncertain harms and benefits that closely resemble what our model suggests are optimal under probability weighting. Because the attractiveness of these policies is difficult to explain otherwise, we present them as inductive evidence that probability weighting may be playing a meaningful role in guiding societal preferences for certain allocations and in determining the actual distributions of harms and benefits in society.

\hhcomment{The last phrase ``playing a role'' may suggest that we consider probability weighting as a mechanism that plays out in people's minds... Can we somehow tone this down?} \hhcomment{I think the following paragraph is mostly left over from the initial sketch of the intro, so it reads a bit abrupt. Can the better writers in the group re-write it please?} \klcomment{I merged the two paras that were above these commments as they were making similar points. Hoda -- feel free to edit, or if you're satisfied can delete this block of comments!}


\xhdr{Prior work on probability weighting}
%
Several models in economics aim to capture people's preferences regarding choices that have uncertain outcomes.
Expected utility theory posits that people are expected-utility maximizers and that they weight probabilities linearly. However, empirical evidence suggests people often do not treat probabilities linearly (see, e.g.,~\citep{quattrone1988contrasting,etchart2004probability,humphrey2004probability,berns2008nonlinear}). Instead, they overweight small probabilities and underweight large probabilities. 
Major alternatives to the expected utility theory, including rank dependent utility~\citep{quiggin1982theory,quiggin2012generalized}, prospect theory~\citep{kahneman2013prospect}, and cumulative prospect theory~\citep{tversky1992advances}, propose the concept of a probability weighting function to capture this behavioral effect. (Note that all these alternatives share the assumption of \emph{additive separability} across outcomes with the expected utility theory.)
A probability weighting function is a widely-studied model of probability distortions in decision making under risk. 
A large number of probability weighting functions have been proposed (see, e.g., \citep{gonzalez1999shape} and \citep{tversky1992advances}).
\citet{prelec1998probability} observes that unlike utility functions, which are characterized by concavity, in empirical studies probability weighting functions are:
\begin{itemize}
\item \textbf{regressive} intersecting the diagonal from above,
\item \textbf{asymmetric} with fixed point at about 1/3,
\item \textbf{s-shaped} concave on an initial interval and convex beyond that,
\item \textbf{reflective} assigning equal weight to a given loss-probability as to a given gain-probability.
\end{itemize}
\citeauthor{prelec1998probability} uses the above observation to axiomatize a subproportional function, $w(p) = \exp{ - ( - \ln p)^\alpha )}$, $0 < \alpha < 1$, that satisfies all four of the above properties, and that has an invariant fixed point and inflection point at $p = 1/e = .37$.
Probability weighting has been featured in various domains, including stock market and the pricing of financial securities (see, e.g.,~\citep{barberis2008stocks}).

In what follows, we build on this prior work, showing that we can derive theoretically what would be the optimal allocation of harms and benefits under probability weighting---allocations that occasionally defy intuition---and we point to a number of concrete cases where we seem to observe such policies in practice.

\section{Behavioral Model}\label{sec:model}
 Consider a society $S$ consisting of $n$ individuals, denoted by $S=\{1,2, \cdots, n\} = [n]$. A policymaker needs to choose a policy $\pi$ that probabilistically allocates some notion of benefit/burden to individuals in $S$. Let $B$ denote the set of all possible levels of benefit/burden that an individual in $S$ can receive. Unless otherwise specified, for simplicity we assume 
$B= \{0,1\}$, with $1$ indicating benefit/burden and $0$ indicating the absence of it. 

A policy $\pi$ distributes a \emph{fixed} level of benefit/harm across individuals $S$. We use the notation $p^{\pi}_i \in [0,1]$ to refer to the probability of benefit/harm imposed on a particular individual $i \in S$ through policy $\pi$. (When $\pi$ is clear from the context, we drop the superscript $\pi$ and use the simplified notation $p_i$ to indicate the probability of individual $i$ receiving the benefit/harm.) For any feasible/admissible policy $\pi$, we assume
\begin{equation}
\sum_{i \in S} p_i = r,
\end{equation}
where $r$ captures the expected amount of benefit/harm that must be distributed.  
(Note that we consider settings in which $r$ is fixed (or changes negligibly) in the number of individuals it is allocated across.)

Following prospect theory, we assume for every individual $i$, there exists an inverse S-shaped function $w_i: [0,1] \rightarrow [0,1]$,
such that $w_i(p)$ determines individual $i$'s perception of probability $p \in [0,1]$. Throughout for simplicity, we assume $w(.)$ and its derivative $w'(.)$ are continuous. 
A concrete instance of a widely-studied probability weighting function is the following \citep{prelec1998probability}:
\begin{equation}
w(p) = e^{-\beta (-\ln{p})^\alpha)}.
\end{equation}
See Figure~\ref{fig:wp} for $\beta=0.5$ and various levels of $\alpha$. 
\begin{figure}
\centering
\includegraphics[width=0.45\textwidth]{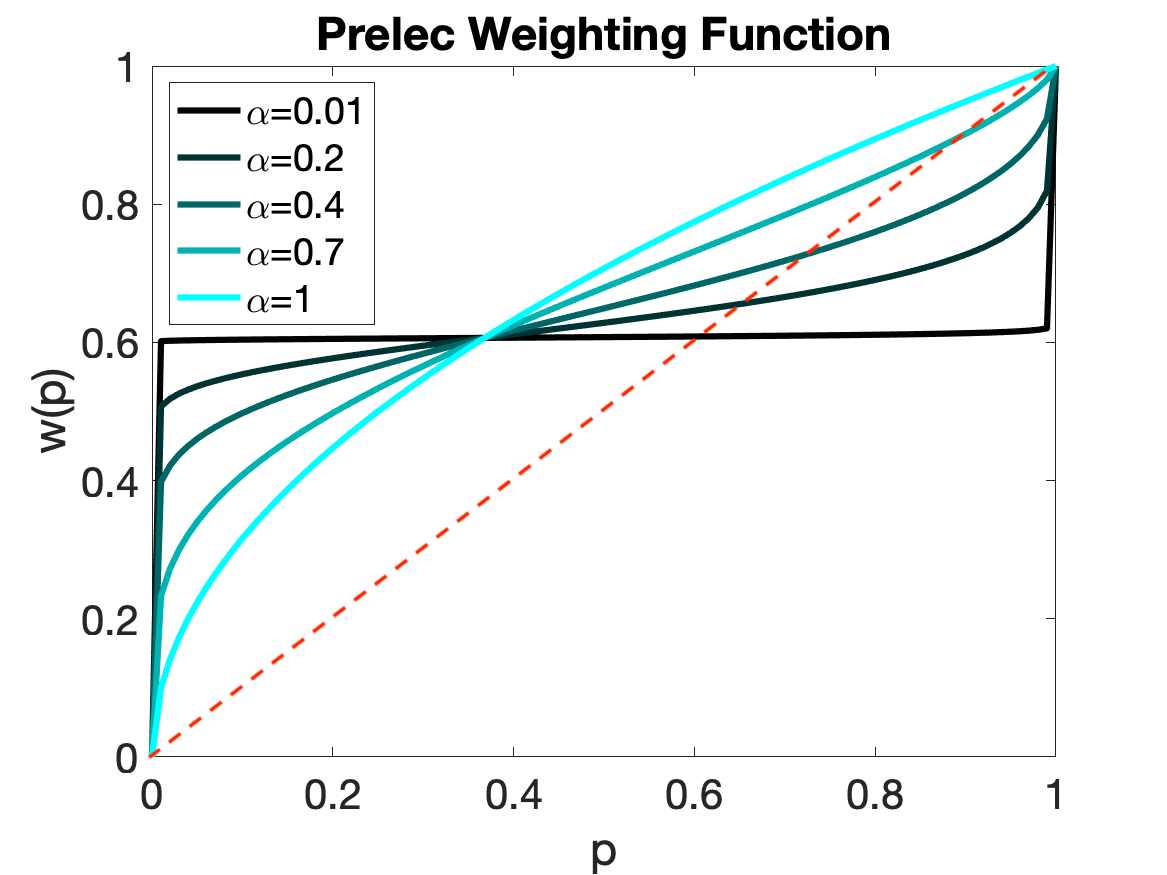}
\caption{Prelec's probability weighting function for $\beta=0.5$}
\label{fig:wp}
\end{figure}
The parameter $\alpha$ determines the \emph{curvature} and $\beta$ determines the \emph{elevation} of the probability weighting function. 


\xhdr{Perceived social welfare} 
Given a policy $\pi$, let $\vp^\pi = (p_1, p_2, \cdots, p_n)$ denote  the distribution of benefit/harm across individuals in $S$ through the policy $\pi$.
We assume that an individual $i$'s judgment of policy $\pi$ can be captured by a score function, $\sigma_i: [0,1]^{n} \rightarrow \mathbb{R}$, which maps $\vp^\pi$ to a real number. The score indicates $i$'s perception regarding the overall benefit/burden that policy $\pi$ imposes on the society $S$. In particular, we assume the score function is additively separable across individuals and is defined as follows:
$$\sigma_i(\vp) = \sum_{j \in S} t_{ji} \times w_i(p_j).$$ 
In the above, $t_{ji}$ indicate the level of \emph{priority} that individual $i$ allocates to $j$ in the context of the allocation problem at hand. For example, the priority $t_{ji}$ may represents $i$'s perception of $j$'s desert/need when allocating benefit, or $j$'s desert/ability to bear the harm when allocating harms. Without loss of generality, we assume $\sum_{j \in S} t_{ji} = 1$ for all $i \in S$ to normalize subjective priorities.   

The overall \emph{perceived welfare} of a policy $\pi$ is calculated by averaging the score function across all individuals, and it is equal to $\sum_{i \in S} \sum_{j \in S} t_{ji} \times w_i(p_j).$ \hhcomment{is this a good place to mentioned that instead of social welfare, one could define the objective in terms of inequality (or a combination of welfare and inequality)?}
For simplicity, throughout we assume that the probability weighting function is the same for all individuals. Under this assumption, the perceived welfare simplifies to $\sigma (\vt, \vp) = \sum_{j \in S} t_j \times w(p_j)$,
where $t_j = \sum_{i \in S} t_{ji}$ is the overall priority of individual $j$ and $\vt=(t_1, \cdots, t_n)$. 
Unless otherwise specified, throughout our analysis we also assume that for all $i,j$, $t_{ji} = 1/n$, that is, priorities are all equal. With this assumption, the perceived welfare simplified to $\sigma(\vp) = \sum_{j \in S} w(p_j)$.

In order to design a policy that probability-weighting individuals perceive positively, 
the policymaker aims to understand the distribution that optimizes the perceived welfare: 
\begin{equation}\label{opt:welfare}
\opt_{\vp \in [0,1]^n}  \sum_{i \in S} w(p_i) \text{ s.t. } \sum_{i \in S} p_i = r.
\end{equation}
In allocating harms, the policymaker is naturally interested in allocations that \emph{minimize} $\sum_{i \in S} w(p_i)$. Conversely, when allocating benefits, he/she is interested in allocations that \emph{maximize} $\sum_{i \in S} w(p_i)$.

In the next two sections, we characterize solution(s) to (\ref{opt:welfare}). Our findings show that perceptions toward harm vs. benefit distributions are fundamentally different.
When allocating \emph{benefit} among $n$ individuals, the perceived-benefit-maximizing solution consists of providing every individual with a non-zero chance of obtaining the benefit.
In contrast, when allocating \emph{harms}, the equal allocation of risk is sub-optimal, and the perceived-harm-minimizing allocation instead  concentrates the risk across $k<n$ individuals leaving the rest to enjoy $0$ risk of harm. 



\section{Perceived Harm Minimizing Allocations}\label{sec:harm_analysis}
We begin our analysis with the case of allocating $r$ units of \emph{harm} among $n$ individuals, and characterize the perceived harm minimizing distributions. We will show that for any given total harm level $r$, the optimal solution to (\ref{opt:welfare}) concentrates the risk on a subset of the population, such that each member of the at-risk subset has a probability of harm bounded away from 1, while most of the population has a probability of harm equal to 0.

Our analysis utilizes the following two lemmas. All omitted proofs can be found in the technical appendix.
\begin{lemma}\label{lem:min_concave}
Let $f: [0,\ell] \longrightarrow [0,1]$ be a strictly concave function, $m \geq 2$ an integer, and $0< c \leq m \ell$ a constant. Let $\vx^* = (\vx^*_1,\cdots,\vx^*_m)$ specify the unique optimal solution to the following minimization problem:
\begin{equation}\label{eq:min_concave}
\min_{x_1,\cdots,x_m \in [0,\ell]} \sum_{i=1}^m f(x_i) \text{ s.t. } \sum_{i=1}^m x_i = c.
\end{equation}
Then for at most one $i \in \{1, \cdots, m\}$, $x^*_i \in (0, \ell)$.
\end{lemma}

\begin{lemma}\label{lem:min_convex}
Let $f: [\ell,1] \longrightarrow [0,1]$ be a strictly convex function, $m \geq 2$ an integer, and $c \in [m \ell, m]$ a constant. Then $(c/m, \cdots, c/m)$ is the unique optimal solution to the following minimization problem:
\begin{equation}\label{eq:min_convex}
\min_{x_1,\cdots,x_m \in [\ell,1]} \sum_{i=1}^m f(x_i) \text{ s.t. } \sum_{i=1}^m x_i = c.
\end{equation}
\end{lemma}

Armed with the above two lemmas, we can characterize the perceived harm minimizing allocations of $r$ units of harm among $n$ individuals.

\begin{theorem}\label{thm:min_harm}
Consider the allocation of $r$ units of harm among $n$ individuals. Let $\ell$ specify the inflection point of the probability weighting function, $w(.)$. Let $\vp^*$ be a minimizer of (\ref{opt:welfare}). Then
\begin{enumerate}
\item For at most one individual $i \in S$, $0<p^*_i<\ell$;
\item For any $j \in S$ with $p^*_j \geq \ell$, $p^*_j = p$ where $p \geq \ell$ is a constant.
\end{enumerate}
\end{theorem}
\begin{proof}
We first show that in the perceived harm minimizing allocation, at most one individual receives $0<p^*_i <\ell$. Suppose not and there are at least two individuals (say $1$ and $2$) such that $0<p^*_1 \leq p^*_2 < \ell$. Next, we apply Lemma~\ref{lem:min_concave} to $w:[0, \ell] \longrightarrow [0,1]$, $m=2$, and $c = p^*_1 + p^*_2$ to argue that there exist an operation to improve the objective value in (\ref{opt:welfare}). Let $(p'_1,p'_2)$ specify the unique optimal solution to the following minimization program:
\begin{equation}
\min_{p_1,p_2 \in [0,\ell]} w(p_1) + w(p_2) \text{ s.t. } p_1 + p_2 = p^*_1 + p^*_2.
\end{equation}
Then according to Lemma~\ref{lem:min_concave}, for at most one $i \in \{1,2\}$, $p'_i \in (0, \ell)$, otherwise $w(p_1) + w(p_2)$ can be improved. 
If there exists an operation to improve $w(p_1) + w(p_2)$, the same operation can be utilized to improve (\ref{opt:welfare}) keeping $p^*_3, \cdots, p^*_n$ (and subsequently, $w(p^*_3), \cdots, w(p^*_n)$ in  (\ref{opt:welfare}) ) unchanged. This would contradicts our initial assumption that for at least two individuals $0<p^*_1 \leq p^*_2 < \ell$. 

Second, we show that  for all $j \in S$ with $p^*_j \geq \ell$, $p^*_j = p$ where $p \geq \ell$ is a constant. Suppose not, and there are two individuals with unequal probabilities in the $[\ell,1]$ interval. Without loss of generality, we assume $\ell \leq p^*_1 < p^*_2$. Applying Lemma~\ref{lem:min_convex} to $w:[\ell,1] \longrightarrow [0,1]$, $m=2$, and $c = p^*_1 + p^*_2$, we can improve the objective value of (\ref{opt:welfare}) by redistributing the probability of harm $(p^*_1 + p^*_2)$ among individuals 1,2 as follows: set both $p_1 = p_2 = (p^*_1 + p^*_2)/2$. Since this operation reduces the objective value, it contradicts the assumption that there are two individuals with unequal probabilities in $[\ell,1].$ 
\end{proof}


Theorem~\ref{thm:min_harm} implies that for any given total harm level $r$, the number of individuals who receive a positive probability of harm is always bounded regardless of the number of individuals $n$. 
Let $k(n,r)$ specify the number of individuals whose probability of harm is greater than or equal to $\ell$ through the optimal solution, $\vp^*$. (More precisely, $k(n,r) = \vert \{i \in S \text{ s.t. } p^*_i = p \geq \ell \}\vert$). 
\begin{corollary}
For any given total harm level $r$, there exists $K \in \mathbb{N}$ such that $k(n,r) \leq K$ for any $n \in \mathbb{N}$. 
\end{corollary}
It is also easy to see that for a fixed $n$, $k(n,r)$ must increase roughly linearly in $r$. This is simply because according to the above theorem, there is at most one individual with $0<p^*_i<\ell$. And when the remaining $(r-p^*_i)$ units of harm is allocated equally among $k(n,r)$ individuals, it leads each one of them to receive a probability harm $p \geq \ell$. Since $p = \frac{(r-p^*_i)}{k(n,r)}$ and $0 \leq p^*_i < \ell$, we have
{\footnotesize
\begin{eqnarray*}
&&\frac{r - \ell}{k(n,r)} \leq p \leq \frac{r}{k(n,r)} \\
&\Rightarrow & \frac{r-\ell}{p} \leq k(n,r) \leq \frac{r}{p} \\
(\ell \leq p \leq 1) &\Rightarrow & r-\ell\leq k(n,r) \leq \frac{r}{\ell}
\end{eqnarray*}
}
So $k(n,r)$ is sandwiched between two linear functions of $r$. Figure~\ref{fig:k-r} demonstrates $k$ as a function of $r$ for $n=50$ when $w$ is the Prelec probability weighting function.

Through our simulations, we observe that for settings depicted in Figure~\ref{fig:k-r}, there exist no individual with $p^*_i < \ell$, so $r$ is divided equally among $k(r)$ individuals. Based on Figure~\ref{fig:k-r}, we conjecture that $k(r) = \floor{c\times r}$ for a constant $c$. In what follows, we derive the slope $c$ for $w(p) = e^{-\beta (-\ln{p})^\alpha)}$.

For a given $r$, the objective value (\ref{opt:welfare}) associated with dividing $r$ equally between $x$ people is equal to $x w(r/x)$. While semantically, $x$ has to be an integer, we can define the function $g(r,x) = x w(r/x)$ over real numbers to facilitate optimization (over a continuous domain of $x$).
Now note that $g(r,x)$ is the perspective function of $w: [\ell,1] \longrightarrow [0,1]$. Since $w$ is convex in the feasible region for $r/x$, we know that its perspective $g$ must also convex~\citep{boyd2004convex}. So for a fixed value of $r$, there exists a unique minimizer $x^*$ for $g(r,x) = x w(r/x)$. 
Approximating $k(r)$ with $c \times r$, we can write $g(r, k) = g(r, cr) = cr e^{-\beta (\log{c})^\alpha}$. Define $h(r,c) := cr e^{-\beta (\log{c})^\alpha}$. We can derive $c$ by setting $ \frac{\partial g}{\partial k}$ to 0. Note that
$$\frac{\partial g}{\partial k} = \frac{\partial g}{\partial h}  \frac{\partial h}{\partial c} \frac{\partial c}{\partial k}.$$
Also, it is easy to observe that $\frac{\partial g}{\partial h} =1$ and $\frac{\partial c}{\partial k} = 1/r > 0$. So it must be the case that $\frac{\partial h}{\partial c} =0.$
{\footnotesize
\begin{eqnarray*}
\frac{\partial h}{\partial c} &=&  r e^{-\beta (\log{c})^\alpha} + c r e^{-\beta (\log{c})^\alpha} \times (-\beta) \times \alpha \times \frac{1}{c} \times (\log{c})^{\alpha-1} \\ 
&=& r e^{-\beta (\log{c})^\alpha} \left( 1 - \beta \alpha (\log{c})^{\alpha-1} \right)   = 0
\end{eqnarray*}
}
From the last equation, we obtain $c = (\alpha \beta)^{\frac{1}{1-\alpha}}$.

\begin{figure}
\centering
\includegraphics[width=0.45\textwidth]{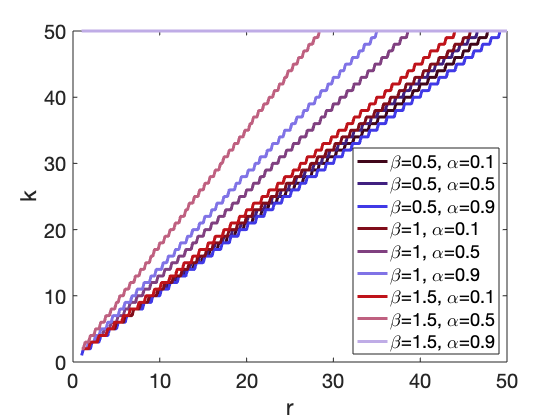}
\caption{The optimal $k$ for various values of $r$ ($n=50$). $k$ is roughly a linear function of $r$ with a slope of $(\alpha \beta)^{\frac{1}{1-\alpha}}$. }
\label{fig:k-r}
\end{figure}

As another concrete example, we can derive the perceived harm minimizing policy for the special case of $r=1$ and the Prelec probability function:
\begin{theorem}\label{thm:min-1}
Consider the distribution $r=1$ units of harm among $n$ individuals. Suppose $w(.)$ is the Prelec function with $\beta=1$. For any $\alpha<1$ and $n>1$, $(0.5, 0.5, 0, \cdots, 0)$ is the unique minimizer of (\ref{opt:welfare}).
\end{theorem}
\begin{proof}
The proof is by induction. The base holds for $n \leq 5$.
(To verify the base, first note that we can have at most two $i$'s with $p_i$'s in the convex region because $3\times 1/e > 1 = r$. A local improvement argument establishes that we can have at most one individual with an allocated probability of $x \leq 1-2/e$ in the concave region. It is easy to see that for $x \in (0, 1-2/e]$, $w(x) + 2 w\left( 0.5(1-x) \right)$ is increasing in $x$, which establishes that the minimum happens at $x=0$). 

The induction hypothesis is that the statement holds for a value of $n \geq 5$.
Next, we show that it holds for $n>5$ as follows: Let vector $\vp$ be the optimal solution to (\ref{opt:welfare}). If there exists a component $0 \leq i \leq n$ with $p_i=0$, then we can apply the induction hypothesis to the remaining $(n-1)$ individuals and we are done.

Otherwise, pick the component $i$ for which $p_i$ is minimum. We must have $p_i  \leq  1/n$ because all $p$'s should add up to 1. We must also have at least one $j \neq i$ for which $p_j  \leq  (1-p_i)/(n-1)$. (Otherwise, the sum of all $p$'s would be greater than $(n-1) \times (1-p_i)/(n-1) + p_i = 1$, which is a contradiction. 

Now if we set the $i$'th component to 0 and $j$'th component to $p_j + p_i$, we have reduced the objective value $w(p_1)+ \cdots + w(p_n)$. This is because 
$p_i + p_j  \leq   p_i + (1-p_i)/(n-1)  \leq  2/n$
(it is easy to verify that $p_i + (1-p_i)/(n-1)$ is increasing in $p_i$ for any $n \geq 3$, so $p_i + (1-p_i)/(n-1)$ is maximized at $p_i = 1/n$ (the upper bound on $p_i$) and it amounts to $2/n$ there).

Now observe that for $n \geq 6$, $2/n < 1/e$.\footnote{Note that for $\beta=1$, $\ell = 1/e$ regardless of the value of $\alpha$.} So when increasing $p_j$ to $p_j + p_i$ and reducing $p_i$ to 0, we remain entirely in the concave region of $w$. So this action reduces the sum of weighted probabilities at $i$'th and $j$'th components. More precisely, because $w$ is concave in this region, we have:
\begin{eqnarray*}
&& w(p_i) - w(0)  \geq  w(p_i + p_j) - w(p_j)\\
&\Leftrightarrow  & w(p_i) + w(p_j) > w(p_i + p_j) + w(0)
\end{eqnarray*}
\end{proof}

\subsection{Heterogeneous Priorities}

Next, we focus on the case of heterogeneous priorities (i.e., non-uniform $\vt$), and characterize the optimal solution(s) when individuals have various priority levels.
Recall that the optimization problem with heterogeneous priorities can be written as follows:
\begin{equation}\label{opt:welfare_heter}
\opt_{\vp \in [0,1]^n}  \sum_{i \in S} t_i \times w(p_i)\text{ s.t. } \sum_{i \in S} p_i = r.
\end{equation}
 where $t_i = \sum_{j \in S} t_{ij}$ is the overall priority of individual $i$. We will use the notation $\sigma (\vt, \vp)$ to refer to the objective function of (\ref{opt:welfare_heter}), and $C$ to refer to all feasible solutions to it, i.e., $C = \{\vp \in [0,1]^n \vert \sum_{i=1}^n p_i = r \}$. Note that $C$ does not vary with $\vt$. 
Let $\sigma^*(\vt)$ be the value function of (\ref{opt:welfare_heter}), and $C^*(\vt)$ its set of optimal solutions for parameter $\vt$. Since $C$ is constant and therefore, continuous (i.e., both upper and lower hemicontinuous) in $\vt$, we can apply the Maximum theorem~\citep{berge1997topological} to establish that the correspondence $C^*(\vt)$ is upper-hemicontinuous in $\vt$. 
\begin{corollary}
Let $\sigma^*(\vt)$ be the value function of (\ref{opt:welfare_heter}), and $C^*(\vt)$ its set of optimal solutions. Then $\sigma^*$ is continuous and $C^*$ is upper hemicontinuous in $\vt$ with nonempty and compact values.
\end{corollary}
In particular, the correspondence is upper-hemicontinuous at $\vt = \ones$ (the case of homogeneous priorities). Therefore, the optimal solution for values of $\vt$ sufficiently close to $\ones$ are close to the optimal solutions at $\ones$ -- which we have already characterized. 

But for values of $\vt$ that are \emph{not} sufficiently close to $\ones$, does the support and shape of the optimal solution substantially change? The answer is no, as long as $\vt \ggcurly \zeros$ (that is, for all $i \in S$, $t_i>0$).

\begin{proposition}\label{prop:min_heter}
Consider the distribution of $r$ units of harm among $n$ individuals with priorities specified through vector $\vt$. For any $\vt \ggcurly \zeros$ and any $\vp^* \in C^*(\vt)$,
\begin{enumerate}
\item There exists at most one individual $i$ with $0<p^*_i<\ell$;
\item For all $j \in S$ with $p^*_j \geq \ell$, $t_j \times w'(p^*_j) = c$ where $c$ is a constant.
\end{enumerate}
\end{proposition}

Let $k(\vt,r)$ specify the number of individuals whose probability of harm is greater than or equal to $\ell$ through an optimal solution, $\vp^* \in C^*(\vt)$. Given the above Proposition and with a reasoning similar to the one provided previously for uniform priorities, it is easy to see that $r-\ell\leq k(\vt,r) \leq \frac{r}{\ell}$. 

\subsection{Examples and Discussion}\label{sec:harm_examples}


\hl{When policymakers make risk allocation decisions in the real world, the resulting policies sometimes resemble the perceived harm-minimizing allocations that our theory predicts. In this Section, in addition to the already-discussed example of military draft, we present several examples of policies---allocating significant, but potentially diffuse harms---that we would struggle to explain without recourse to probability weighting. We offer these real-world examples \emph{not} to claim that policymakers do--or should--explicitly take the effects of probability weighting on people's preferences into account when making policy decisions, nor that the policies we detail here are, empirically, the most common ways policymakers allocate risk. But we offer them here as suggestive evidence that probability weighting may play a role---even an implicit one---in 
leading policymakers or the public to perceive certain outcomes as more attractive in real-world settings, and is therefore important to consider.}

\hl{
The concrete examples that we review here are unmistakably unpleasant to think about, and the reader may find it morbid to dwell on them. Yet the distastefulness of these contexts is precisely the reason why we suggest that probability weighting may play a role in them. The fact that these decisions involve such extreme harms may explain why decision-makers might rely, implicitly, on psychological distortions that give the impression that the chosen allocation has helped to reduce the total amount of harm.
}

\xhdr{Environmental pollutants} 
A recurring setting in which the public is confronted with uncertain risks of harm is in the effect of environmental pollutants on people's health. We draw on this setting as a stylized example for probability weighting by adapting a hypothetical scenario from Robert Sapolsky on the diffusion of risk \cite{sapolsky1994measures}. In his scenario, Sapolsky asks us to consider the psychological impact of dangers from pollution concentrated in a small area versus pollution that is spread more diffusely over a large area. We will see that probability weighting provides a potential way of thinking about some of the different reactions we have to contrasting scenarios. 

In particular, consider the following three hypothetical scenarios involving harms from pollution. In Scenario 1, the public learns that a company has decided to unsafely dispose of a certain amount of toxic waste by burying it on a remote plot of land next to a house where 3 people live far from the rest of the population; the concentration of the chemicals is so high that each of these three people acquires a rare, deadly disease that has a typical prevalence in the population of essentially 0. In Scenario 2, we learn that a company has decided to unsafely dispose of this same amount of toxic waste by spreading it diffusely over a larger region where a community of people live, in the process increasing the disease risk by some intermediate level such that three people will die of the same rare disease in expectation. In Scenario 3, we learn that a company has, again, disposed of this same amount of toxic waste by venting it into the air over a metropolitan area where three million people live, giving each of them a .0001\% chance of acquiring this rare disease. (Again, we suppose this chance is significantly higher than the base-rate of this disease in the population, which is effectively zero.) 

We begin by positing that the public might well have different subjective reactions to these stories if they were presented in isolation. At a minimum, the three stories have very different structures. The first evokes the image of a specific set of three people who have died as a direct result of the company's actions; they are the {\em identifiable victims} in the story \citep{jenni1997explaining}. The second evokes the image of the community that is involved, and the people who are now at elevated risk of the disease. And in the third, the company's actions have had an impact on millions of people, all of whom have been endangered by its actions. 

These different subjective framings of the scenarios should be contrasted with the fact that, in each case, the company's recklessness leads to three deaths from the same disease in expectation. This suggests that something other than the expected number of deaths is leading to the subjective differences. Moreover, the differences are also hard to fit in any simple way into an equity principle, since in all three scenarios, most people in society are completely unaffected by the harm from this particular pollution incident, and so the contrasts are more about the number of people affected and how seriously than about any sense in which all of society is sharing the harm equally. 

Probability weighting provides a framework for thinking about the contrasts. If $k$ people are impacted by the pollution, each receiving a probability of $r/k$ of contracting the resulting disease, then the total perceived harm from this allocation of risk, under the model in this section, is $k w(r/k)$. Even though the total harm $r$ is fixed in the three scenarios, the total perceived harm $k w(r/k)$ changes as we vary $k$, and so under probability weighting, the total perceived harm will appear different across the scenarios. Moreover, Theorem~\ref{thm:min_harm} shows that for some choices of these parameters, it is possible to view the intermediate scenario as producing the lowest perceived harm---when the harm is diffuse enough that we neither have identifiable victims (as in Scenario 1), nor that we produce positive risk for too large a population (as in Scenario 3). 

Beyond the specifics of the formal model, however, our purpose in this example is also to highlight a broader qualitative point rooted in the contrasts drawn by Sapolsky's example \citep{sapolsky1994measures}: that probability weighting can lead us to have very different reactions to a fixed amount of expected harm, and that our reactions can be non-monotone in the sense that harm to an intermediate-sized group can, in principle, generate less discomfort than harm to a group that is either too small or too large.

\jkdelete{Consider an \hhcomment{hypothetical? stylized?} \sbcomment{Ah---yes, we should add hypothetical. We shoudl also maybe give some justification for using a hypothetical when the point is to offer real-world examples. Would it suffice to say that this is an example contemplated in the existing literature and so we're extending it slightly? Or should we say that this example feels quite realistic?} example that we adapt from \citep{sapolsky1994measures}: a villainous industrialist, contemplating a profitable venture that will generate toxic waste as a byproduct of the manufacturing process for his widgets, needs to decide how he will dispose of the dangerous material. Confronted with such a decision, the industrialist can choose to spread the risk of exposure diffusely across the local population, perhaps dumping the waste in the nearby waterway. Alternatively, he might concentrate the risk by burying the waste in one discrete location. Or he might place the waste in a set of smaller bins that he then distributes across a few different sites. In all three cases, the total amount of harm due to exposure remains roughly the same, but each choice would distribute the risk of harm differently.

Given that the same number of people are harmed in expectation across all three choices, the industrialist should have no particular preferences between his options. In \citeauthor{sapolsky1994measures}'s version of the story, the industrialist is told that ``the toxins his factory will dump into the drinking water will probably lead to three cancer deaths in his town of 100,000.'' But rather than contemplating the harm that his plans would cause to three specific people, the industrialist re-frames his choices in terms of probabilities: ``All it does is increase the cancer risk .003 percent for each person.'' \citeauthor{sapolsky1994measures} focuses on how probabilistic thinking can insulate the industrialist from having to grapple with concrete harms, but it also points toward a behavioral bias captured by our model. Spreading a risk diffusely can help to reduce the \textit{perceived} harm, even when the harm remains constant. On \citeauthor{sapolsky1994measures}'s account, the more diffusely the industrialist is able to spread the risk, the less guilt he will feel. But spreading the risk too broadly could also engender a more violent reaction from the local population than the reaction to a more intermediate distribution. At the extreme, like dumping the waste in the local waterway, the entire population might be put at some small risk, upsetting \textit{everyone} who now faces a non-zero chance of harm.

It is easy to imagine others confronting similar dilemmas, even those with a greater sense of responsibility than the industrialist in \citeauthor{sapolsky1994measures}'s story. How should they go about making their decision?

\hhdelete{
The first step is to 

 is to choose a value of $k$ and spread the risk over $k$ people. Next, we discuss some of the alternatives available. 
 
\begin{enumerate}
\item Choose $k = n$: They could spread the toxic waste very thinly over the entire country. If caught, the public would say, ``You put the entire country at risk.'' They'd argue, ``But think how small the risk to each of you was.''
\item Choose $k = r$: They could find a house where $k$ people live and bury it directly next to them. All $k$ people die, but no one else is exposed to the risk. If caught, the public would say, ``You killed these $k$ people.'' They'd argue, ``But at least the rest of the public was safe.''
\item Choose an intermediate value of $k$ between $r$ and $n$: They choose a medium-sized town and dump it in an adjacent body of water. $r$ extra people in the town die, but the rest of the country isn't exposed to risk. 
\end{enumerate} 
Intuitively, a company caught in situation 3 might well be vilified less than a company caught in situation 1 or (certainly) in situation 2.  (We see companies caught in situation 3 on a regular basis \hh{citations?}.)  

How can we describe our intuitive reactions to the alternatives in the above hypothetical example?
Our model offers one way of deriving from first principles the idea that the minimum perceived harm is to spread the risk over a set of ``medium'' size rather than over a set that is too large or too small.}
\hhcomment{Do we want to keep a version of the following sentence -- just to point at the specific result we are referring to?}
Theorem~\ref{thm:min_harm} suggests that from a probability weighting perspective, one prefers solutions where the size of the at-risk subgroup, $k$, lies somewhere between the two extremes of $r$ and $n$. 
\sbcomment{I definitely think it makes sense to draw directly on the theorem to help answer this question. I wonder if it's actually enough to just add that one sentence to the end of the previous paragraph and be done with the entire example.}
}

\xhdr{The executioner's conscience.}
Governments that enact capital punishment have sometimes done so with significant attention to the moral guilt felt by executioners. (There is, of course, no small degree of dark irony in a policy aimed at alleviating the pain of killing somebody.) Military firing squads traditionally involved at least one rifle carrying a blank cartridge---called the ``conscience round''---distributed at random, so that all members of the squad could avoid knowing with certainty that their own actions directly led to the death of the victim \citep{sapolsky1994measures,brennan2019strange}. Some executions by lethal injection have carried forward a similar practice. Lethal injection machinery used for a time in some states involved dual sets of syringes and dual switches to activate them, which were to be thrown simultaneously by two people. A computer would randomize which of the two vials would be injected into the prisoner and which would be discarded, and would then erase the determination---therefore giving both executioners a means of avoiding the certainty of knowing whether they, individually, had delivered the injection to the prisoner \citep{sapolsky1994measures}.

While acknowledging the perversity of a policy that focuses on the pain an execution causes to an executioner, we detail this strategy because it closely fits the situation we model, as follows. In a dual execution team, the ``harm''---again, here construed as the moral guilt of having killed a person---is distributed according to probabilities $(0.5,0.5,0,0,...,0)$ (as predicted in Theorem~\ref{thm:min-1}). Why? What is the advantage of orchestrating the execution in such a way, which certainly is more complicated and costly than simply omitting the placebo? With a single executioner, the operator knows without doubt that they performed the execution, which presumably causes them psychic harm. But with dual teams---the solution we know to be optimal according to Theorem~\ref{thm:min-1}---we distribute this harm over \emph{two} people. Both executioners are spared the moral weight of certainty---after all, there's only a 50\% chance each was the killer. And in fact, if we were to randomize the harm over \emph{more} people, we would end up with a sub-optimal allocation from a perceived harm perspective (even putting aside the monetary costs of additional redundancy in the system). As we know from probability weighting, three people bearing a 33\% chance of harm is perceived to be more harmful than two people each bearing 50\%.

This situation is also interesting because, in many real-world instances of harm allocation, probabilities are assigned---but then at some future point, the uncertainty is resolved and it becomes clear who actually experienced the harm. In contrast, in this case, the uncertainty is (deliberately) never resolved---by design, the computer discards the determination. The cost that agent $i$ experiences even after the event is the weighted probability $w(p_i)$ that they caused the death. Therefore, the total perceived harm remains at $w(p_1) + w(p_2) + ... + w(p_n)$ both before and after the deed. This perceived harm is precisely the objective function we study.

\section{Perceived Benefit Maximizing Allocations}\label{sec:benefit_analysis}

Next, we analyze the case of allocating $r$ units of \emph{benefit} among $n$ individuals, and characterize the perceived benefit maximizing allocations.
We will show that the number of individuals who receive a non-zero probability of benefit is always $n$ regardless of the size of benefit, $r$. In addition and perhaps surprisingly, for a fixed value of $n$, if the total benefit level $r$ is sufficiently large relative to $n$, the optimal solution allocates probabilities in a two-tier manner where a subset of individuals receive the benefit with certainty and the remaining individuals have equal (but less than 1) chances of obtaining the remaining benefit.

Our analysis utilizes the following two lemmas. All omitted proofs can be found in the technical appendix.
\begin{lemma}\label{lem:max_concave}
Let $f: [0,\ell] \longrightarrow [0,1]$ be any strictly concave function, $m \geq 2$ an integer, and $0< c \leq m \ell$ a constant. Then $(c/m, \cdots, c/m)$ is the unique optimal solution to the following maximization problem:
\begin{equation}\label{eq:max_concave}
\max_{x_1,\cdots,x_m \in [0,\ell]} \sum_{i=1}^m f(x_i) \text{ s.t. } \sum_{i=1}^m x_i = c.
\end{equation}
\end{lemma}

\begin{lemma}\label{lem:max_convex}
Let $f: [\ell,1] \longrightarrow [0,1]$ be any strictly convex function, and $m \ell < c \leq m$ a constant. Let $\vx^* = (\vx^*_1,\cdots,\vx^*_m)$ specify the unique optimal solution to  the following maximization problem:
\begin{equation}\label{eq:max_convex}
\max_{x_1,\cdots,x_m \in [\ell,1]} \sum_{i=1}^m f(x_i) \text{ s.t. } \sum_{i=1}^m x_i = c.
\end{equation}
Then for at most one $i \in \{1, \cdots, m\}$, $x^*_i \in (\ell,1)$.
\end{lemma}

Armed with the above two lemmas, we can characterize the perceived benefit maximizing allocation of $r$ units of benefit among $n$ individuals.

\begin{theorem}\label{thm:max_benefit}
Consider the distribution of $r$ units of benefit among $n$ individuals. Let $\ell$ specify the inflection point of the probability weighting function, $w(.)$. Let $\vp^*$ be a maximizer of (\ref{opt:welfare}). Then
\begin{enumerate}
\item For at most one individual $i$, $\ell < p^*_i < 1$;
\item For all $j \in S$ with $0 \leq p^*_j \leq \ell$, $p^*_j = p$ where $0 \leq p \leq \ell$ is a constant.
\end{enumerate}
\end{theorem}
\begin{proof}
We first show that in the perceived benefit maximizing allocation, at most one individual receives $\ell < p^*_i < 1$. Suppose not and there are at least two individuals (say 1 and 2) such that $\ell<p^*_1 \leq p^*_2 < 1$. According to Lemma~\ref{lem:max_convex}, we can improve the objective value by redistributing the probability of benefit $(p^*_1 + p^*_2)$ among individuals 1,2 as follows: we reduce $p_1$ and increase $p_2$--keeping $p_1 + p_2$ constant at $(p^*_1 + p^*_2)$--until either $p_1$ reaches $\ell$ or $p_2$ reaches $1$. Note that this operation improves the objective value and in either case it contradicts our initial assumption that for at least two individuals $\ell <p^*_1 \leq p^*_2 < 1$. 

Second, we show that for all $j \in S$ with $0 \leq p^*_j \leq \ell$, $p^*_j = p$ where $0 \leq p \leq \ell$ is a constant. Suppose not, and there are two individuals with unequal positive probabilities in the $[0,\ell]$ interval. Without loss of generality, we assume $0\leq p^*_1 < p^*_2 \leq \ell$. According to Lemma~\ref{lem:max_concave}, we can improve the objective value by redistributing the probability of harm $(p^*_1 + p^*_2)$ among individuals 1,2 as follows: set both $p_1 = p_2 = (p^*_1 + p^*_2)/2$. Since this operation improves the objective value, it contradicts the assumption that there are two individuals with unequal probabilities in $[0,\ell].$
\end{proof}

Next and as a concrete example, we focus on the special case of the Prelec probability weighting functions with $\beta=1$, and derive the perceived benefit maximizing policy for $r=1$. According to theorem~\ref{thm:max_benefit}, it is easy to observe the following:

\begin{corollary}\label{cor:max_1}
Consider the distribution $r=1$ units of benefit among $n$ individuals. Suppose $w(.)$ is the Prelec function with $\beta=1$. For any $\alpha<1$ and $n>1$, the unique maximizer of (\ref{opt:welfare}) is either
\begin{itemize}
    \item the uniform allocation $(1/n,\cdots,1/n)$;
    \item or $(\epsilon,...,\epsilon, 1-(n-1)\epsilon)$ for some $\epsilon \leq \min\{1/e, \frac{1-\ell}{n-1}\}$.
\end{itemize}
\end{corollary}\label{thm:max-1}
%

Moreover, we can show that for sufficiently large $n$, $(1/n,\cdots,1/n)$ is the unique maximizer. 

\begin{theorem}\label{thm:no_epsilon}
Consider the distribution $r=1$ units of benefit among $n$ individuals. Suppose $w(.)$ is the Prelec function with $\beta=1$ and $\alpha<1$.
Then there exists a value $q(\alpha,\beta) \in \mathbb{N}$ such that for any $n > q(\alpha,\beta)$, $(1/n,\cdots,1/n)$ is the unique global maximizer of (\ref{opt:welfare}).
\end{theorem}
\begin{proof}
Note that since $w$ is concave up to the inflection point $\frac{1}{e}$, its derivative is positive and decreasing in $p$. Also we know that the derivative amount to $+\infty$ at $p=0$, so there exists a value $0 < q < \frac{1}{e}$ at which $w'(q) = 1$.
We can choose $n$ large enough such that $\frac{1}{n-1} < q$. We next show that for such values of $n$, no distribution of the form $\left(\frac{\epsilon}{n-1}, \cdots, \frac{\epsilon}{n-1}, 1-\epsilon \right)$ with $\epsilon < 1-\frac{1}{e}$ can be optimal.

The proof is by contradiction. Suppose not and $\left(\frac{\epsilon}{n-1}, \cdots, \frac{\epsilon}{n-1}, 1-\epsilon \right)$ is a maximizer of (\ref{opt:welfare}) for some $\epsilon < 1-\frac{1}{e}$. Next, we show that we can strictly improve the objective by moving from this distribution to $(\frac{1}{n-1}, \cdots, \frac{1}{n-1}, 0)$. This is a contradiction with the assumption that $\left(\frac{\epsilon}{n-1}, \cdots, \frac{\epsilon}{n-1}, 1-\epsilon \right)$ is a maximizer.
To establish the above claim, we can write:
{\scriptsize
\begin{eqnarray*}
\sigma\left(\frac{\epsilon}{n-1}, \cdots, \frac{\epsilon}{n-1}, 1-\epsilon \right) &=& w(1-\epsilon) +  (n-1) w\left(\frac{\epsilon}{n-1}\right) \\
&\leq & 1-\epsilon +  (n-1) w\left(\frac{\epsilon}{n-1}\right) \\
&\leq & (n-1) \left[ \frac{1-\epsilon}{n-1} +  w\left(\frac{\epsilon}{n-1}\right) \right] \\
&\leq & (n-1) \left[ w\left(\frac{1}{n-1}\right) - w\left(\frac{\epsilon}{n-1}\right) +  w\left(\frac{\epsilon}{n-1}\right) \right] \\
&\leq & (n-1) w\left(\frac{1}{n-1}\right)
\end{eqnarray*}
}
In the second line of the above derivation, we utilized the fact that $1-\epsilon > \frac{1}{e}$--the fixed point and point of inflection of $w(.)$. So $w(1-\epsilon) < 1-\epsilon$. 
In the third line, we utilized the fact that $w\left(\frac{1}{n-1}\right) - w\left(\frac{\epsilon}{n-1}\right) > \frac{1-\epsilon}{n-1}$. To see this, note that
{\scriptsize
\begin{equation*}
w\left(\frac{1}{n-1}\right) - w\left(\frac{\epsilon}{n-1}\right) =  \int_{\frac{\epsilon}{n-1}}^{\frac{1}{n-1}}  w'(p) dp 
> \int_{\frac{\epsilon}{n-1}}^{\frac{1}{n-1}}  1 dp 
= \frac{1-\epsilon}{n-1}.
\end{equation*}
}
\end{proof}

\xhdr{Conditions for two-tier optimal allocations} According to Theorem~\ref{thm:max_benefit}, when $r>1$, the perceived benefit maximizing solution can consist of individuals who receive the benefit with certainty (i.e., $p^*_i = 1$). Next we ask: how large should $r$ be relative to $n$ for at least one individual to receive the benefit with certainty in the perceived benefit maximizing allocation? The following Proposition establishes that the necessary and sufficient condition is $r = \Theta(n)$.
\begin{proposition}\label{prop:certain_benefit}
If $r \geq (n-1)\ell + 1$, there exists an individual $i \in S$ with $p^*_i = 1$. Moreover, there exists a constant $q$ such that if $r \leq q \times n $, $p^*_i < 1$ for all $i$. 
\end{proposition}

Figure~\ref{fig:min_r} depicts the minimum $r$ required for at least one individual to receive the benefit with certainty as a function of $n$ when $w(.)$ is the Prelec probability weighting function with $\alpha=0.9$ and $\beta=1$.
\begin{figure}
\centering
\includegraphics[width=0.45\textwidth]{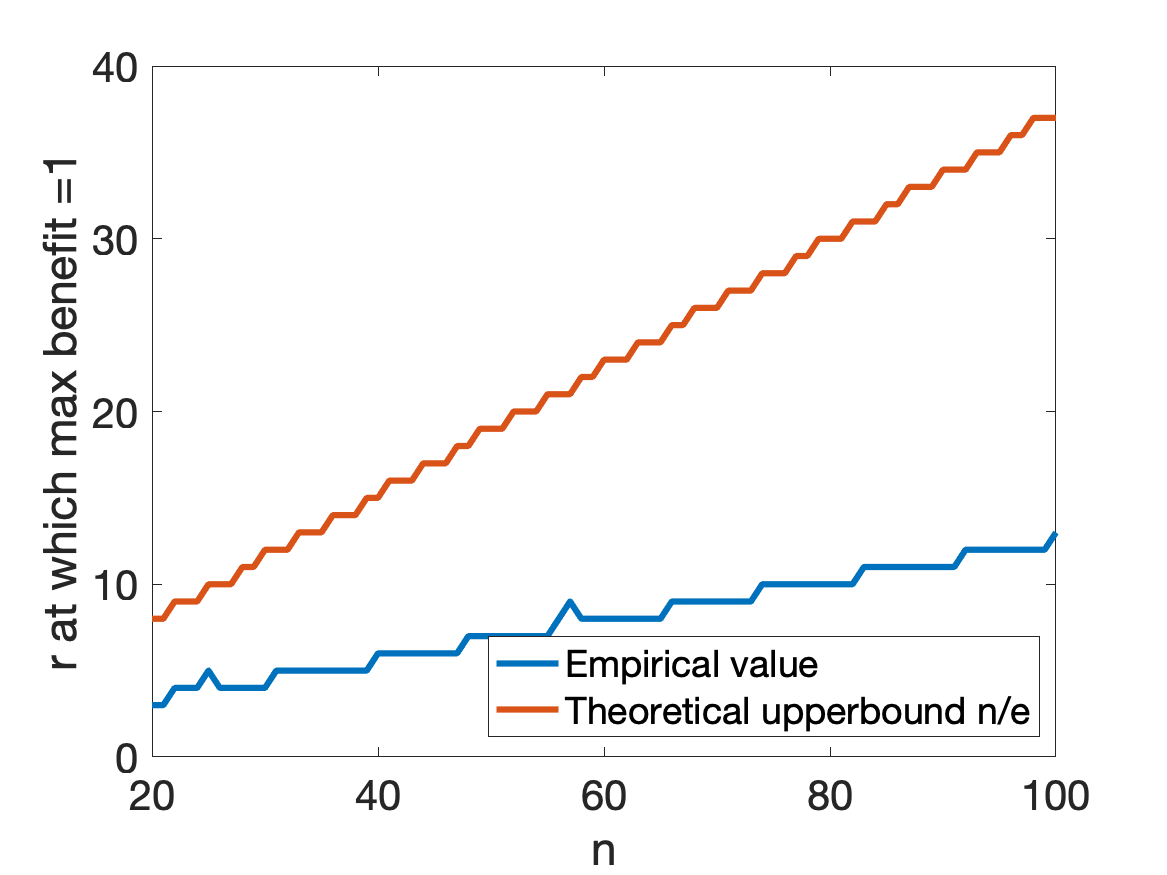}
\caption{The minimum $r$ required for at least one individual to receive the benefit with certainty when the probability weighting function is the Prelec function with $\alpha=0.9$ and $\beta=1$.}
\label{fig:min_r}
\end{figure}

\xhdr{Heterogeneous priorities} Similar to the case of allocating harms, a heterogeneous priority vector $\vt$ does not significantly impact the shape and support of the perceived benefit maximizing allocations, as long as for all $i \in S$, $t_i>0$. 
\begin{proposition}\label{prop:max_heter}
Consider the distribution of $r$ units of benefit among $n$ individuals with priorities specified by $\vt$. 
For any $\vt \ggcurly \zeros$ and any $\vp^* \in C^*(\vt)$, $\vp^*  \ggcurly \zeros$. 
\end{proposition}

\subsection{Examples and Discussion}\label{sec:benefit_examples} 

As with the case of distributions of harm, it is useful to interpret the theoretical results for distributions of benefit by asking how they reflect real-world scenarios. As in our previous discussion of examples, we emphasize again that in any real-world instance of benefit allocation, numerous factors impact the ultimate distribution, and we do not claim that probability weighting is the only factor shaping the benefit allocations in the examples that follow. Instead, we argue that the instances presented here collectively serve as inductive evidence for probability weighting playing \emph{a} possible role in how benefits are distributed in society, and of policymakers grappling with choices between distributions of uncertain benefits. 

\xhdr{Uniform lotteries} We begin by observing that when the total benefit $r$ is small relative to the size of the population under consideration ($n$), the allocation that maximizes the perceived benefit is a {\em uniform lottery} in which everybody has an equal chance of receiving the benefit. Examples of uniform lotteries are widespread; what is interesting is that our model produces this natural outcome without an overt preference for \emph{equalizing} probabilities among people. Rather, the equalization of probabilities is arising for a different reason: because probability weighting inflates small probabilities, we find that assigning everyone a small but equal probability of receiving the benefit serves to maximize the sum of weighted probabilities. 
In this way, an objective that maximizes total perceived welfare produces an outcome that indirectly optimizes for a certain type of equity in the probabilities as well. 

Note also that this forms a strong contrast with the case of harm distributions. For probabilities of harm, this same principle from probability weighting---that small probabilities are inflated---implies that uniform lotteries will produce an unnecessarily {\em large} perceived total cost to the population. Hence, for harms, the optimal solution was to run a lottery over a much smaller subset of the population. This reflects an important contrast between uncertain harms and uncertain benefits as viewed through the lens of probability weighting: we may prefer to spread a small possibility of benefit very widely across the population, but would view the same distribution as producing unacceptable levels of total perceived risk when it is used for allocating harms. 

\xhdr{Systems with discrete tiers} Once the total benefit $r$ is sufficiently large (relative to the size of the population $n$), the maximum perceived distribution takes on a very different form: it provides benefit to a subset of the population with probability $1$, and runs a uniform lottery on the remainder. (Potentially there may also be one intermediate value as well.) This form of the solution, in which benefit is assigned in a small number of discrete ``tiers of service,'' has a more unusual structure; informally, the probability weighting model discovers that it can \hhdelete{provide significant gain}\hhedit{significantly improve the perceived benefit} to a subset of the population by pushing their probabilities up to $1$, while reducing the perceived benefit to the rest of the population only very slightly when it shifts their probabilities correspondingly downward to balance the total. 

Despite this unusual structure, we can find a number of cases in practice where probabilistic allocations of benefit use structures based on a small number of discrete tiers. As above, we do not claim that this arises directly from explicit probability-weighting considerations, but that a recurring preference for such solutions suggests that this type of approach---which is difficult to motivate via other simple probabilistic models---reflects a consistency with the underlying principles of probability weighting. 

One paradigmatic example of this multi-tiered approach is in the allocation of hunting permits. For safety and conservation reasons, states often impose limits on how many of a given animal species may be hunted in a season. There are often more people who want to hunt a particular type of animal than there are animals that the state deems appropriate to be hunted---requiring states to determine how to allocate this benefit among them. States approach this problem by running ``tag lotteries'' in which each tag grants the right to harvest one animal. States differ in how they structure these lotteries. Some conduct uniform lotteries in which each entrant has equal odds of obtaining a tag, while others structure them to reward longtime entrants who have not successfully drawn tags in previous rounds. They may do so by, for example, giving hunters an additional entry for each year they have been unsuccessful, increasing their chances---though not ensuring with certainty---that they receive the benefit in the current round. Still others let hunters accumulate ``preference points'' for each unsuccessful year, and then \emph{guarantee} tags to all hunters who meet some point threshold, while allocating any remaining tags by random lottery \citep{huntsplus}---effectively creating a two-tiered structure of benefit allocation, like that predicted by our model (Theorem~\ref{thm:max_benefit}). 

When states design these allocation mechanisms, they find themselves wrestling explicitly with how to distribute uncertain benefits fairly and how to balance between certainty and randomness. For example, in the late 2000s, Montana---which had used a preference point system to reward more senior hunters, to the degree that the process had become essentially a seniority queue in which hunters could wait decades for some prized tags---reformed its system to reserve a small number of tags for random lottery, in order to incentivize new, young hunters to participate \citep{montana}. Through this two-tier system, Montana gave many longtime hunters the certainty of knowing they'd be able to hunt; but by injecting a modest amount of uncertainty into the system via a second random tier, the state aimed to encourage others to behave based on a belief that they, too, might partake of the benefit.


Allocation systems with a small number of tiers, each with different probabilities, arise in other cases where permission to enter a restricted activity is being granted; for example, entries in the New York City marathon follow a similar high-level structure, with a portion of the entries allocated based on deterministic criteria and the rest allocated by lottery \citep{marathon}. Upgrade policies in brand loyalty programs can be viewed as setting different probabilities of receiving benefits for a small number of different tiers as well, though of course the specifics of each policy can become complex.
\hhcomment{Quick comment on ``small number of tiers'': our model only justifies two tiers, and with very particular structure to those two tiers. So its connection to settings with small number of tiers is not immediately clear.}
An intriguing feature of all these examples is that we typically think of the use of tiers as a way of controlling the cognitive complexity of a policy: it is easier for people to understand a small number of categories than a continuum of different probability values. In contrast, the model based on probability weighting has no intrinsic reason to group the allocation probabilities into a small number of tiers, since any distribution is an allowable option for it; the fact that it nevertheless creates small numbers of discrete tiers for its optimal solutions suggests a fundamental connection between this type of tiered discretization and the structure of preferences based on probability weighting.

\section{Conclusion and Future Directions}\label{sec:conclusion}

\newcommand{\omt}[1]{}

We have considered policies that distribute probabilities of harm,
or probabilities of benefit, across a population, and have asked
how we might distinguish among policies that produce different distributions
with the same total expected impact.
The theory of probability weighting from behavioral economics provides
one natural proposal: since people systematically perceive probabilities 
to be different from their actual values, 
two distributions with the same expectation will not in general
have the same sum of {\em weighted} probabilities.
Accordingly, we have investigated which types of distributions
optimize certain functions of these perceived (weighted) probabilities,
and have found that distributions with the characteristics of these
optima show up in a diverse range of policy contexts.
Our point in this analysis is not to recommend such distributions
as being normatively preferable to others, but instead to argue
that a societal preference for them is consistent with (and hence captured well by) basic principles from the behavioral science of
probability perception.

Our results also reveal a number of opportunities for further research, both empirical and theoretical. 

As we've emphasized throughout, we do not expect that our model fully accounts for the diverse considerations that inform decisions regarding the allocation of probabilistic harms and benefits. By design, our model highlights one aspect of a complex process that might help account for the preferences and policies that we seem to observe in practice. But far more empirical work is necessary to establish what role probability weighting actually plays in any given policymaking process. Our model points to particular types of allocation policies that are worthy of further investigation. Such work could help uncover empirical details that support alternative plausible explanations, but it could also find that policymakers struggle to offer a coherent justification for their chosen allocation, perhaps lending support to the idea that their judgments might rest, implicitly, on distorted perceptions of probabilities. Empirical methods from behavioral economics would no doubt be useful in this work as well, testing how various stakeholders involved in the policymaking process actually perceive the relevant probabilities of harm and benefit.

We also observe that the types of distributions favored by models based on probability weighting can be seen to align with principles of distributive justice in some cases (particularly when probabilities are distributed more uniformly), and to contrast with these principles in other cases (such as real-world instances where policies choose to concentrate risk on smaller sets, for example). It would be interesting to search for deeper foundations that might underpin these similarities and contrasts.

Since our analysis is stylized, we should be able to bring similar thinking to bear in other domains and extend our reasoning beyond the types of cases we've considered here. For the sake of concreteness, we'll discuss these extensions in terms of harm rather than benefit, but the questions generally apply in both cases.

To start, we are modeling cases in which there is a single kind of harm (engaging in a dangerous activity or not, being exposed to a hazardous environment or not) and the risk is experienced once. It would clearly be of interest to incorporate multiple levels of harm or harms that evolve over time. Moving to this non-binary setting would require that we incorporate behavioral models that account for ranked levels of utility into probability weighting (e.g., by distinguishing between relative losses and gains), as well as behavioral biases around costs incurred in the present versus the future. All of these might shed further light on settings where these types of policy questions arise.

Probability weighting is also not the only systematic behavioral bias that people exhibit when dealing with probabilities. For example, the phenomenon of {\em base rate neglect} leads people to overweight the evidence of a single instance rather than correctly taking into account the background distribution that the instance comes from. People also systematically display {\em overconfidence} when estimating the probability of beneficial outcomes that are based on their own agency \citep{della-vigna-psych-econ}. It would be interesting to see how these might be integrated into a framework for analyzing distributional questions as we do here. \klcomment{citation or two here?} \sbcomment{Do you happen to have some go-to citations, Hoda? If not, I can try to dig some up!}\jkcomment{I added a citation to a survey paper.}

Of course, none of these behavioral effects tell us how to choose among policies offering different distributions, but they can provide insight, in some cases, into why people express the policy preferences that they do. \hl{By integrating consideration of a robust behavioral bias like probability weighting into these contexts, we may better understand policy preferences for different mechanisms with uncertain allocations.}

\omt{
\begin{itemize}

    \item Settings in which the total harm $r$ is a function of $k$ 
    \item Multiple levels of harm/benefit (instead of the simplified version where $B = \{0,1\}$). This would necessitate taking loss aversion into account, and define probability weighting in a rank-dependent manner.
    \item One-off vs. repeated allocations of risk
    \item Other behavioral biases involving probabilities (e.g., probability neglect, probability matching?)
    \item Taking advantage of uncertainty in order to influence people to behave in certain ways
    \item  Behavioral economics as a field is devoted to uncovering systematic cognitive biases that explain apparent departures from rational thinking. From this perspective, it’s easy to see probability weighting as a discovery that allows us to explain preferences that don’t, at first, seem like they are the result of a rational cost-benefit analysis. But once we recognize the distorting effects of probability weighting, then it’s possible to understand these allocation preferences as the result of a rational cost-benefit analysis that just happens to be based on distorted expected costs and benefits. In other words, it lets us salvage cost-benefit analysis as an explanatory tool. But there’s something else going on here: probability weighing also has the effect of attuning us to questions of distributive justice---precisely the questions that cost-benefit analysis often fails to consider. There’s something intuitively much less offensive about two people having 50 chance of suffering a harm than concentrating certain harm on one person---and this is likely because we don’t see the total cost as the \textit{same} in these two situations. In a way, probability weighting performs a similar function as the “distributional weights” that philosophers and economists have proposed as a way to integrate distributional concerns into cost-benefit analysis. Probability weighting makes it rational to favor distributions of harms and benefits that are not highly unequal because doing so seems to reduce the overall harm or increase the overall benefit.
\end{itemize}

}



\section*{Acknowledgments}
We are grateful to Ted O'Donoghue,  Sendhil Mullainathan,  the participants of the AI,  Policy,  and Practice (AIPP) group at Cornell,  and the FEAT reading group at Carnegie Mellon for valuable discussions.  This material is based upon work supported in part by NSF IIS2040929,
a Simons Investigator Award,
a Vannevar Bush Faculty Fellowship,
a MURI grant,  AFOSR grant FA9550-19-1-0183,  and
grants from the ARO and the John D.  and Catherine T.  MacArthur Foundation.
Any opinions, findings and conclusions or recommendations expressed in this material are those of the authors and do not necessarily reflect the views of the National Science Foundation and other funding agencies.

{\small
\bibliographystyle{plainnat}
\bibliography{biblio_PW}
}

\pagebreak
\section{Remaining Proofs from the Main Text}\label{app:technical}


\subsection{Proof of Lemma~\ref{lem:min_concave}}
\begin{proof}
Suppose the statement does not hold, and there exist $i,j \in \{1, \cdots, m\}$ with $i \neq j$ and $x^*_i, x^*_j \in (0, \ell)$. Without loss of generality, we assume $0 <x^*_1 \leq x^*_2 < \ell$. Next, we construct $\vx'$, another feasible solution to (\ref{eq:min_concave}), from $\vx^*$ and show that $\vx'$ leads to a lower objective value compared to $\vx^*$, that is, $\sum_{i=1}^m f(x'_i) < \sum_{i=1}^m f(x^*_i)$. This observation would be a contradiction with the optimality of $\vx^*$.

We construct $\vx'$ as follows: If $x^*_1 + x^*_2 \leq \ell$, replace $x^*_1$ and $x^*_2$ in $\vx'$ with $0$ and $(x^*_1 + x^*_2)$, respectively. Otherwise and if $x^*_1 + x^*_2 > \ell$, replace them with $(x^*_1 + x^*_2) - \ell$ and $\ell$. It is trivial to verify that in both cases, $\vx'$ remains a feasible solution to (\ref{eq:min_concave}).
Next we prove that $\vx'$ improves the objective function. (We provide the proof for the case in which $x^*_1 + x^*_2 \leq \ell$, but the proof is identical for the case of $x^*_1 + x^*_2 > \ell$) 

First, note that $\sum_{i=1}^m f(x'_i) = \sum_{i=1}^m f(x^*_i) - f(x^*_1) - f(x^*_2) + f(x^*_1 + x^*_2) + f(0) $, so it suffices to show that $- f(x^*_1) - f(x^*_2) + f(x^*_1 + x^*_2) + f(0) < 0$ or equivalently, $f(x^*_1 + x^*_2) + f(0) < f(x^*_1) + f(x^*_2)$. To prove this, note that we can write $f(x^*_1) + f(x^*_2)$ as follows:
{\scriptsize
\begin{eqnarray*}
 &= & f\left( (x^*_1 + x^*_2)\frac{x^*_1}{x^*_1 + x^*_2} \right) + f\left( (x^*_1 + x^*_2)\frac{x^*_2}{x^*_1 + x^*_2} \right) \\
  &= & f\left( (x^*_1 + x^*_2)\frac{x^*_1}{x^*_1 + x^*_2} + 0\frac{x^*_2}{x^*_1 + x^*_2} \right) + f\left( 0\frac{x^*_1}{x^*_1 + x^*_2} + (x^*_1 + x^*_2)\frac{x^*_2}{x^*_1 + x^*_2} \right)
\end{eqnarray*}
}
Applying the definition of concavity to the two terms in the equation above, we obtain that:
{\scriptsize
\begin{eqnarray*}
&>& \frac{x^*_1}{x^*_1 + x^*_2}  f( x^*_1 + x^*_2) + \frac{x^*_2}{x^*_1 + x^*_2} f(0) +  \frac{x^*_1}{x^*_1 + x^*_2}f(0) + \frac{x^*_2}{x^*_1 + x^*_2} f( x^*_1 + x^*_2) \\
& = & f( x^*_1 + x^*_2) + f(0)
\end{eqnarray*}
}
or equivalently, $f(x^*_1) + f(x^*_2) > f (x^*_1 + x^*_2) + f(0).$ Therefore, we have that $\sum_{i=1}^m f(x^*_i) > \sum_{i=1}^m f(x'_i)$.
\end{proof}

\subsection{Proof of Lemma~\ref{lem:min_convex}}
\begin{proof}
Let $\vx^* = (\vx^*_1,\cdots,\vx^*_m)$ specify the optimal solution to (\ref{eq:min_convex}).
Suppose the statement does not hold, and $\vx^*$ contains two unequal numbers. Without loss of generality, suppose $x^*_1 \neq x^*_2$ with $\ell \leq x^*_1 < x^*_2 \leq 1$. Let $\vx'$ be another potential solution obtained from $\vx^*$ by replacing $x^*_1$ and $x^*_2$ in $\vx^*$ with $\frac{x^*_1 + x^*_2}{2}$ in $\vx'$. Next, we show that $\vx'$ is a feasible solution to (\ref{eq:min_convex}) and $\sum_{i=1}^m f(x'_i) < \sum_{i=1}^m f(x^*_i)$. This observation would be a contradiction with the optimality of $\vx^*$.

To show that $\vx'$ is feasible, note that since $x^*_1, x^*_2 \in [\ell,1]$, their average is also in $[\ell,1]$. Also, note that $x'_1 + x'_2  = \frac{x^*_1 + x^*_2}{2} + \frac{x^*_1 + x^*_2}{2} = x^*_1 + x^*_2$, so $\sum_{i=1}^m x'_i = \sum_{i=1}^m x^*_i = c$. 
To prove that $\vx'$ improves the objective function in (\ref{eq:min_convex}), note that $f$ is strictly convex, so using the definition of convexity, we know that:
$$\frac{1}{2} f(x^*_1) + \frac{1}{2}f(x^*_2)  > f \left(\frac{1}{2} x^*_1 + \frac{1}{2}x^*_2 \right),$$
or equivalently, $f(x^*_1) + f(x^*_2) > 2 f \left(\frac{x^*_1 + x^*_2}{2} \right) = f(x'_1) + f(x'_2).$ Therefore, we have that $\sum_{i=1}^m f(x^*_i) > \sum_{i=1}^m f(x'_i)$.
\end{proof}

\subsection{Proof of Proposition~\ref{prop:min_heter}}
\begin{proof}
We first show that in the optimal solution to (\ref{opt:welfare_heter}), at most one individual receives $0<p^*_i <\ell$. Suppose not and there exist two individuals $1,2$ with $0<p^*_1 \leq p^*_2<\ell$. It must be the case that $t_1 > t_2$, otherwise we could reduce the objective value by swapping $p^*_1$ and $p^*_2$. Next we show that we can improve the objective value by redistributing the probability of harm $(p^*_1 + p^*_2)$ among individuals 1,2 as follows: we reduce $p_1$ and increase $p_2$--keeping $p_1 + p_2$ constant at $(p^*_1 + p^*_2)$--until either $p_1$ reaches $0$ or $p_2$ reaches $\ell$. Note that this operation strictly improves the objective value because: (1) according to Lemma~\ref{lem:min_concave} the reduction in $w(p_1)$ is larger than the increase in $w(p_j)$; (b) the reduction in $w(p_1)$ is multiplied by $t_1$, while the increase in $w(p_2)$ is multiplied by $t_2$. This local improvement contradicts our initial assumption that in an optimal allocation, there can exist two individuals with $0<p^*_1 \leq p^*_2 < \ell$. 

Second, we show that for all $j \in S$ with $p^*_j \geq \ell$, $t_j \times w'(p^*_j) = c$ for a constant $c$. Suppose not, and there are two individuals (say 1,2) with assigned probabilities in the $[\ell,1]$ interval such that $t_1 w'(p^*_1) \neq t_2 w'(p^*_2)$. Without loss of generality, we assume $t_1 w'(p^*_1) < t_2 w'(p^*_2)$. Now note that we can improve the objective value by redistributing the probability of harm $(p^*_1 + p^*_2)$ among individuals 1,2 by solving the following \emph{convex} optimization problem:
$$\min_{\ell < p_1,p_2 < 1} t_1 w(p_1) + t_2 w(p_2) \text{ s.t. } p_1 + p_2 = p^*_1 + p^*_2.$$
We can equivalently write the above as
$$\min_{p_1,p_2 \in \reals} t_1 \tilde{w}(p_1) + t_2 \tilde{w}(p_2) \text{ s.t. } p_1 + p_2 = p^*_1 + p^*_2.$$
where $\tilde{w}(.)$ denotes the extended-value extension of $w: [\ell,1] \longrightarrow [0,1]$. 
Writing the stationarity conditions for the above equivalent problem, we obtain:
\begin{eqnarray*}
t_1 \tilde{w}'(p_1) = t_2 \tilde{w}'(p_2) = c \\
&\Rightarrow& t_1 w'(p_1) = t_2 w'(p_2) = c
\end{eqnarray*}
which is a contradiction with our initial assumption.
\end{proof}

\subsection{Proof of Lemma~\ref{lem:max_concave}}
\begin{proof}
The proof is identical to Lemma~\ref{lem:min_convex}. The only difference is that the optimization operator is maximization (instead of minimization), and concavity changes the direction of the inequalities in the proof.
\end{proof}

\subsection{Proof of Lemma~\ref{lem:max_convex}}
\begin{proof}
The proof is identical to Lemma~\ref{lem:min_concave}. The only difference is that the optimization operator is maximization (instead of minimization), and concavity changes the direction of the inequalities in the proof.
\end{proof}

\subsection{Proof of Theorem~\ref{thm:no_epsilon}}
\begin{proof}
Note that since $w$ is concave up to the inflection point $\frac{1}{e}$, its derivative is positive and decreasing in $p$. Also we know that the derivative amount to $+\infty$ at $p=0$, so there exists a value $0 < q < \frac{1}{e}$ at which $w'(q) = 1$.
We can choose $n$ large enough such that $\frac{1}{n-1} < q$. We next show that for such values of $n$, no distribution of the form $\left(\frac{\epsilon}{n-1}, \cdots, \frac{\epsilon}{n-1}, 1-\epsilon \right)$ with $\epsilon < 1-\frac{1}{e}$ can be optimal.

The proof is by contradiction. Suppose not and $\left(\frac{\epsilon}{n-1}, \cdots, \frac{\epsilon}{n-1}, 1-\epsilon \right)$ is a maximizer of (\ref{opt:welfare}) for some $\epsilon < 1-\frac{1}{e}$. Next, we show that we can strictly improve the objective by moving from this distribution to $(\frac{1}{n-1}, \cdots, \frac{1}{n-1}, 0)$. This is a contradiction with the assumption that $\left(\frac{\epsilon}{n-1}, \cdots, \frac{\epsilon}{n-1}, 1-\epsilon \right)$ is a maximizer.
To establish the above claim, we can write:
{\scriptsize
\begin{eqnarray*}
\sigma\left(\frac{\epsilon}{n-1}, \cdots, \frac{\epsilon}{n-1}, 1-\epsilon \right) &=& w(1-\epsilon) +  (n-1) w\left(\frac{\epsilon}{n-1}\right) \\
&\leq & 1-\epsilon +  (n-1) w\left(\frac{\epsilon}{n-1}\right) \\
&\leq & (n-1) \left[ \frac{1-\epsilon}{n-1} +  w\left(\frac{\epsilon}{n-1}\right) \right] \\
&\leq & (n-1) \left[ w\left(\frac{1}{n-1}\right) - w\left(\frac{\epsilon}{n-1}\right) +  w\left(\frac{\epsilon}{n-1}\right) \right] \\
&\leq & (n-1) w\left(\frac{1}{n-1}\right)
\end{eqnarray*}
}
In the second line of the above derivation, we utilized the fact that $1-\epsilon > \frac{1}{e}$--the fixed point and point of inflection of $w(.)$. So $w(1-\epsilon) < 1-\epsilon$. 
In the third line, we utilized the fact that $w\left(\frac{1}{n-1}\right) - w\left(\frac{\epsilon}{n-1}\right) > \frac{1-\epsilon}{n-1}$. To see this, note that
{\scriptsize
\begin{equation*}
w\left(\frac{1}{n-1}\right) - w\left(\frac{\epsilon}{n-1}\right) =  \int_{\frac{\epsilon}{n-1}}^{\frac{1}{n-1}}  w'(p) dp 
> \int_{\frac{\epsilon}{n-1}}^{\frac{1}{n-1}}  1 dp 
= \frac{1-\epsilon}{n-1}.
\end{equation*}
}
\end{proof}

\subsection{Proof of Proposition~\ref{prop:certain_benefit}}
\begin{proof}
We first show that if $r \geq (n-1)\ell + 1$, then at least one individual receives the benefit with certainty.
Suppose not, and for all $i \in S$, $p^*_i < 1$. Combining this fact with theorem~\ref{thm:max_benefit}, we can deduce that for every $i \in S$, one of the following must be the case: (1) $0 \leq p^*_i \leq \ell$, or (2) $\ell< p^*_i<1$. We also know according to theorem~\ref{thm:max_benefit} for at most one individual $i$, $\ell< p^*_i<1$. The maximum total benefit in a distribution with these properties realizes when $(n-1)$ individuals belong to category 1 with $p = \ell$ and one individual belongs to category 2 with $p$ very close but smaller than 1. So the maximum total benefit in such distribution is less than $(n-1)\ell + 1$, which is a contradiction with $r \geq (n-1)\ell + 1$.

Second, we show that there exists a constant $q$ such that if $r \leq q \times n $, no one receives the benefit with certainty. 
Let $0 < q \leq \ell$ be the point at which $w'(.)=1$. We can show that for $r \leq q \times n $, no one will receive the benefit with certainty in the optimal solution $\vp^*$. Suppose not, and there exists $i \in S$ such that $p^*_i = 1$. Let $m \leq (n-1)$ specify the number of individuals with $p^* \leq q$. Also, let $S_m \subset S$ denote the set of those individuals for whom $p^* \leq q$. With a reasoning similar to the one presented in the proof of Theorem~\ref{thm:no_epsilon}, it is easy to observe that the objective value can be strictly improved by taking $1$ unit of benefit from $i$, and distributing it equally among individuals in $S_m \cup \{i\}$. This is contradiction with the initial assumption that there exists an individual $i$ with $p^*_i = 1$.
\end{proof}

\subsection{Proof of Proposition~\ref{prop:max_heter}}
\begin{proof}
Suppose not and there exists an $i$ with $p^*_i=0$. Pick another individual $j$ with $0< p^*_j \leq 1$. 
We have that $w'(p^*_i) = w'(0) = +\infty$ and $w'(p^*_j) < +\infty$, so we can reduce $p^*_j$ by an infinitesimal amount $\epsilon$ and increase $p^*_i$ at the same time, keeping $p^*_i + p^*_j$ constant. If $\epsilon$ is chosen to be small enough, then this operation will improve the perceived benefit, which is a contradiction.
\end{proof}

\end{document}